\documentclass{amsart}
\usepackage{amsmath,amssymb}
\usepackage{amsfonts}
\usepackage{amsthm}
\usepackage{comment}
\newcommand{\jap}[1]{\langle #1 \rangle}

\def\a{\alpha}
\def\b{\beta}

\def\d{\delta}
\def\e{\varepsilon}
\def\f{\varphi}
\def\g{\psi}

\def\k{\kappa}
\def\l{\lambda}
\def\m{\mu}

\def\s{\sigma}

\def\x{\xi}
\def\y{\eta}

\def\H{\mathcal{H}}
\def\Fd{\mathcal{F}_d}

\def\re{\mathbb{R}}
\def\co{\mathbb{C}}
\def\ze{\mathbb{Z}}

\def\T{\mathbb{T}}
\def\pa{\partial}

\renewcommand{\Re}{\text{{\rm Re}\;}}
\renewcommand{\Im}{\text{{\rm Im}\;}}

\newcommand{\supp}{\text{{\rm supp}\;}}

\newcommand{\Ker}{\text{{\rm Ker}\;}}
\newcommand{\Ran}{\text{{\rm Ran}\;}}

\newtheorem{thm}{Theorem}[section]
\newtheorem{lem}[thm]{Lemma}
\newtheorem{prop}[thm]{Proposition}
\newtheorem{cor}[thm]{Corollary}

\theoremstyle{definition}
\newtheorem{defn}{Definition}

\theoremstyle{remark}
\newtheorem{rem}[thm]{Remark}

\title{Uniform bounds of discrete Birman-Schwinger operators}

\author{Yukihide Tadano}
\address{Graduate School of Mathematical Sciences, University of Tokyo, 3-8-1 Komaba, Meguroku, Tokyo, Japan 153-8914}
\email{tadano@ms.u-tokyo.ac.jp}
\author{Kouichi Taira}
\address{Graduate School of Mathematical Sciences, University of Tokyo, 3-8-1 Komaba, Meguroku, Tokyo, Japan 153-8914}
\email{taira@ms.u-tokyo.ac.jp}

\begin{document}

\maketitle

\begin{abstract}
In this note, uniform bounds of the Birman-Schwinger operators in the discrete setting are studied. For uniformly decaying potentials, we obtain the same bound as in the continuous setting. However, for non-uniformly decaying potential, our results are weaker than in the continuous setting. As an application, we obtain unitary equivalence between the discrete Laplacian and the weakly coupled systems.
\end{abstract}

\section{Introduction}

 We consider the discrete Schr\"odinger operators:
\begin{align*}
H=H_0+ V(x)\quad \text{on}\quad \H= l^2(\mathbb{Z}^d),
\end{align*}
where $H_0$ is the negative discrete Laplacian
\begin{align*}
H_0 u(x)=-\sum_{|x-y|=1}(u(y)-u(x)),
\end{align*}
and $V$ is a real-valued function on $\ze^d$.  In this note, we study uniform bounds of the Birman-Schwinger operators:
\begin{align}\label{BS}
\sup_{z\in \co\setminus \re}\left\||V|^{\frac{1}{2}}(H_0-z)^{-1}|V|^{\frac{1}{2}}\right\|_{B(\H)}<\infty.
\end{align}
 As an application, we give sufficient conditions for $V$ that $H_0$ and $H$ are unitarily equivalent. We also give examples of potentials for which (\ref{BS}) does not hold. Though this subject is studied in a recent preprint \cite{BPL}, their assumptions are stronger than ours and some proofs seem incomplete. One of the purposes of this note is to generalize their results and give an alternative proof.

We denote the Fourier expansion by $\Fd$:
\begin{align*}
\hat{u}(\x)=\mathcal{F}_d u(\x)=\sum_{x\in \ze^d}e^{-2\pi i x\cdot\x}u(x),\quad \x\in \mathbb{T}^d=\re^d/\ze^d.
\end{align*}
Then it follows that
\begin{align*}
\Fd H_0 u(\x)=h_0(\x)\Fd u(\x),
\end{align*}
where $h_0(\x)=4\sum_{j=1}^d\sin^{2}(\pi \x_j)$, and hence $\s(H_0)=[0,4d]$. We denote the set of the critical points of $h_0$ by $\Gamma$:
\begin{align*}
\Gamma=\{\x\in \T^d\,|\, \nabla h_0(\x)=0\}=\{\x\in \T^d\,|\, \x_j\in \{0, 1/2\},\, j=1,...,d\}.
\end{align*}
We call $\x\in \Gamma$ an elliptic threshold if $\x$ attains maximum or minimum of $h_0$ and a hyperbolic threshold otherwise.

For a measure space $(X,\m)$, $L^{p,r}(X,\m)$ denotes the Lorentz space for $1\leq p\leq \infty$ and $1\leq r\leq \infty$:
\begin{align*}
&\|f\|_{L^{p,r}(X)}=\begin{cases}
p^{\frac{1}{r}}(\int_0^{\infty}\m(\{x\in X\,|\, |f(x)|>\a\})^{\frac{r}{p}} \a^{r-1}d\a)^{\frac{1}{r}},\quad &r<\infty,\\
\sup_{\a>0}\a\m(\{x\in X\,|\, |f(x)|>\a\})^{\frac{1}{p}},\quad &r=\infty,
\end{cases}\\
&L^{p,r}(X,\m)=\{f:X\to \mathbb{C}\,|\, f:\text{measurable},\, \|f\|_{L^{p,r}(X)}<\infty\}.
\end{align*}
Moreover, we denote $L^{p,r}(\re^d)=L^{p,r}(\re^d, \m_L)$ and $l^{p,r}(\ze^d)=L^{p,r}(\ze^d, \m_c)$, where $\m_L$ is the Lebesgue measure on $\re^d$ and $\m_c$ is the counting measure on $\ze^d$. For a detail, see \cite{G}.

First, we state our positive results:
\begin{thm}\label{mainpo}\begin{itemize}
\item[$(i)$] Let $d\geq 4$. If $V\in l^{\frac{d}{3},\infty}(\ze^d)$, then (\ref{BS}) holds.
\item[$(ii)$] Let $d\geq 3$. If $|V(x)|\leq C(1+|x|)^{-2}$ for some $C>0$, then  (\ref{BS}) holds.
\end{itemize}
\end{thm}

\begin{cor}\label{mainco}
Under the condition of  Theorem \ref{mainpo} $(i)$ or $(ii)$, $H=H_0+\l V$ is unitarily equivalent to $H_0$ for small $\l \in \re$.
\end{cor}

\begin{rem}
For Theorem \ref{mainpo} $(ii)$, we show stronger results in Proposition \ref{discuni2}. For Theorem \ref{mainpo} $(i)$, we also obtain stronger results: Uniform resolvent estimates in Lorentz spaces as in Proposition \ref{discuni}.
\end{rem}

\begin{rem}
In \cite{IK} and \cite{SK}, the authors prove the absence of eigenvalues of $H_0+\l V$ for small $\l\in \re$ if $|V(x)|\leq C(1+|x|)^{-2-\e}$ for some $C>0$ and $\e>0$ with $d=3$ and $V\in l^{\frac{d}{3}}(\ze^d)$ with $d\geq4$ respectively. In \cite{BPL}, (\ref{BS}) is proved under stronger assumptions: $|V(x)|\leq C\jap{x}^{-2(d+3)}$ with $d\geq3$. Moreover, in \cite{KM}, $(\ref{BS})$ is established for $V\in l^{p}(\ze^d)$ with $1\leq p<6/5$ if $d=3$ and $1\leq p<3d/(2d+1)$ if $d\geq 4$. The authors in \cite{KM} also study the scattering theory of $H_0+V$.

\end{rem} 

\begin{rem}
Theorem \ref{mainpo} $(ii)$ holds if $H_0$ is replaced by a Fourier multiplier $\Fd^{-1}e\Fd$, where $e$ is a Morse function on $\mathbb{T}^d$. In fact, any Morse function can be deformed into ultrahyperbolic operators near its critical points. Thus we can apply the arguments in Section 3 directly. On the other hand, the authors are not confident whether we may replace $H_0$ by $\Fd^{-1}e\Fd$ in Theorem \ref{mainpo} $(i)$ due to the difficulty of multidimensional versions of the van der Corput lemma.
\end{rem}

\begin{rem}
Theorem \ref{mainpo} $(ii)$ is optimal as is shown in Theorem \ref{mainne} below. However, the authors expect that Theorem \ref{mainpo} $(i)$ is far from optimal.
\end{rem}

Corollary \ref{mainco} follows from Theorem \ref{mainpo} and the following classical result due to T.  Kato:

\begin{lem}[{\cite[Theorem XIII.26]{RS}\label{Ksmooth}}]
Let $H_0$ be a positive self-adjoint operator on a Hilbert space $\mathcal{H}$ and let $V$ be a bounded self-adjoint operator on $\mathcal{H}$. If 
\begin{align*}
\sup_{z\in \co\setminus \re}\left\||V|^{\frac{1}{2}}(H_0-z)^{-1}|V|^{\frac{1}{2}}\right\|_{B(\H)}<\infty,
\end{align*}
then $H_0$ and $H_0+\l V$ are unitarily equivalent for small $\l \in \re$.
\end{lem}

Moreover, we state the existence of the boundary values of the free resolvent near $\Gamma$:

\begin{thm}\label{bv}
Suppose $s>1$. Then, $\jap{x}^{-s}(H_0-z)^{-1}\jap{x}^{-s}$ is H\"older continuous in $B(\mathcal{H})$ with respect to $z\in \mathbb{C}_{\mp}=\{z\in \mathbb{C}\,|\, \mp\Im z>0\}$. In particular, the incoming/outgoing resolvents
\begin{align*}
\jap{x}^{-s}(H_0-\m\pm i0)^{-1}\jap{x}^{-s}
\end{align*}
exist in the operator norm topology of $B(\mathcal{H})$ for $\m\in [0,4d]$.
\end{thm}

As a corollary, we have upper bounds of the number of discrete eigenvalues of $H_0+\l V + W$ when $W$ is finitely supported.

\begin{cor} \label{cor}
Let $H=H_0+ \l V$, where $V$ satisfies the condition of Theorem \ref{mainpo} $(i)$ or $(ii)$ and $\l \in \re$ is small. Let $W$ be a real-valued finitely supported potential. 
Then
\begin{align*}
\dim \Ker (H+W-\m) \leq \#\left\{ x\in\ze^d\, |\, W(x) \neq 0 \right\}
\end{align*}
for any $\m\in\re$ and
\begin{align*}
&\dim \Ran E_{H+W}^{pp} \left( \left( -\infty, 0 \right] \right) \leq \#\left\{ x\in\ze^d\, |\, W(x) < 0 \right\}, \\
&\dim \Ran E_{H+W}^{pp} \left( \left[ 4d, \infty \right) \right) \leq \#\left\{ x\in\ze^d\, |\, W(x) > 0 \right\} , 
\end{align*}
where $\dim E_{H+W}^{pp} \left( I \right)$ denotes the projection onto the eigenspace of $H+W$ corresponding to the eigenvalues contained in $I \subset \re$.
\end{cor}

\begin{rem}
This corollary appears in \cite[Corollary 2.4]{BPL} under different assumptions in a stronger form. However, their argument seems to be incomplete. Indeed, their proof of the positivity of the quadratic form $(\f, [V,iA] \f)_\mathcal{H}$ on the eigenspace of $H_0+V$ does not work when $H_0+V$ has at least two eigenvalues, where $V$ is a real-valued function with finite support and $A$ is the conjugate operator associated to $H_0$.
\end{rem}

Next, we state our negative results:
\begin{thm}\label{mainne}
\begin{itemize}\item[$(i)$] Suppose $d=1$ or $2$. For a non positive potential $V\in l^{\infty}(\ze^d)$ which is not identically zero and vanishes at infinity, (\ref{BS}) does not hold. Moreover, $H_0+\l V$ has at least one eigenvalue for all $\l>0$.
\item[$(ii)$] Suppose $d=2$. Let $\chi\in C^{\infty}(\T^d)$ be a non-negative function which is equal to $1$ near $\{\x_1,\x_2\in \{1/4, 3/4\}\}$ and is supported near $\{\x_1,\x_2\in \{1/4, 3/4\}\}$.
Then, there exists $w\in l^{2,\infty}(\ze^d)$ such that
\begin{align}\label{ii}
\sup_{z\in\mathbb{C}\setminus \re}\|w\chi(D)(H_0-z)^{-1}w\|_{B(\H)}=\infty.
\end{align}
Moreover, $\chi(D)(H_0-z)^{-1}$ is not uniformly bounded in $B(l^p(\ze^d),l^q(\ze^d))$ for $1\leq p\leq \infty$ and $q<\infty$.
\item[$(iii)$] Let $d\geq 3$ and $q>\frac{d+2}{3}$. Then, there exists $V\in l^{q,\infty}(\ze^d)$ such that (\ref{BS}) does not hold. In particular, if $d\geq 5$ then there exists $V\in l^{\frac{d}{2},\infty}(\ze^d)$ such that (\ref{BS}) does not hold.
\item[$(iv)$] Let $d\geq 3$ and $V(x)=(1+|x|)^{-\a}$ for $0\leq \a<2$. Then (\ref{BS}) does not hold.
\end{itemize}
\end{thm}

\begin{rem} 
Theorem \ref{mainne} $(i)$ for $d=2$ was conjectured in \cite{SK}.
\end{rem}
\begin{rem}
Theorem \ref{mainne} $(i)$ and $(iv)$ hold even when $H_0$ is a Fourier multiplier $\Fd^{-1}e\Fd$ with a Morse function $e$. We expect that $(iii)$ also holds for such operators, however we have no proof for the moment.
\end{rem}

\begin{rem}
The left hand side of (\ref{ii}) is finite for the continuous Schr\"odinger operator $H_0=-\Delta$ on $L^2(\re^2)$, $w\in L^{q,\infty}(\re^2)$ ($2<q\leq 3$), and $\chi\in C_c^{\infty}(\re^d)$ which is supported away from the threshold $0$. For a proof, we use \cite[Theorem 5.8]{R} (uniform resolvent estimates for the two dimensional case), a real interpolation argument, and H\"older's inequality as in the proof of \cite[Corollary 2.3]{Mi}. 
\end{rem}

There are various works concerning bounds of Birman-Schwinger operators for the continuous Schr\"odinger operators (see \cite{KY}, \cite{KRS}, \cite{RS}).  For $H_0=-\Delta$ on $L^2(\re^d)$ with $d\geq 3$, it is known that (\ref{BS}) holds for $V\in L^{\frac{d}{2},\infty}(\re^d)$ (see \cite{KRS}). Moreover, this result is sharp in the sense that (\ref{BS}) does not hold for $V(x)=|x|^{-\frac{d}{q}}\in L^{q,\infty}(\re^d)$ if $q\neq \frac{d}{2}$. In fact, by scaling
\begin{align*}
\||x|^{-\frac{d}{2q}}(-\Delta-z)^{-1}|x|^{-\frac{d}{2q}}\|_{B(\H)}=\e^{2-\frac{d}{q}}\||x|^{-\frac{d}{2q}}(-\Delta-\e^2z)^{-1}|x|^{-\frac{d}{2q}}\|_{B(\H)}
\end{align*}
holds and we consider the limits as $\e\to 0$ and $\e\to \infty$. Cuenin's examples in \cite[Remark 1.9]{C} which are based on the examples by Frank and Simon (\cite{FS}, see also \cite{IJ}) show that there exists a sequence of real-valued potentials $V_n$ which satisfy $|V_n(x)|\leq C(n+|x|)^{-1}$ and induce an embedded eigenvalue of $H_0+V_n$. 
 
We compare our results with the continuous case. For uniformly decaying potentials $V(x)=(1+|x|)^{-\a}$, the range of $\a$ where (\ref{BS}) holds is the same as in the case of continuous Schr\"odinger operators. However, for non-uniformly decaying potentials, for example $V\in l^{p,\infty}(\ze^d)$, the classes of potentials where (\ref{BS}) holds differ between the discrete case and the continuous case. It seems that this is a consequence of the anisotropy of the discrete Laplacian.

Our paper is organized as follows. In section \ref{secul}, in order to study properties of the resolvent of $H_0$ near $\Gamma$, we investigate properties of the ultrahyperbolic operators. In section \ref{secpo}, we prove our positive results Theorems \ref{mainpo}, \ref{bv} and Corollary \ref{cor}. In section \ref{secne}, we give the proofs of our negative results Theorem \ref{mainne}.

We use the following notations throughout this paper. For Banach spaces $X$ and $Y$, $B(X,Y)$ denotes a set of all bounded linear operators from $X$ to $Y$. We denote the norm of a Banach space $X$ by $\|\cdot\|_X$. We also denote $(\cdot,\cdot)_X$ by the inner product of a Hilbert space $X$. Moreover, we write $B(X)=B(X,X)$. We denote $\jap{x}=(1+|x|^2)^{\frac{1}{2}}$ and $D_x=(2\pi i)^{-1}\nabla_x$ for $x\in \re^d$. A symbol $\mathcal{F}$ denotes the Fourier transform on $\re^d$:
\begin{align*}
\mathcal{F}u(\x)=\int_{\re^d}e^{-2\pi ix\cdot \x}u(x)dx,\quad \x\in\re^d.
\end{align*}
For $\chi \in C^{\infty}(\T^d)$ or $\chi\in C^{\infty}_c(\re^d)$, we denote $\chi(D)=\mathcal{F}_d^{-1}\chi\mathcal{F}_d$ or $\chi(D)=\mathcal{F}^{-1}\chi\mathcal{F}$ respectively.
For $h\in C^{\infty}(\T^d)$ or $h\in C^{\infty}(\re^d)$, $a\in\re$ and a compactly supported smooth function $f$, if $\nabla h\neq 0$ on $\{h(\x)=a\}\cap \supp f$, set
\begin{align*}
\d(h(D)-a)f(x)=\int_{\{h=a\}}\hat{f}(\x)e^{2\pi i x\cdot \x}\frac{d\s(\x)}{|\nabla h(\x)|},
\end{align*}
where $d\s$ is the induced surface measure.

We give a useful formula. Let $N\subset \T^d$ or $N\subset \re^d$ be a submanifold which has the following graph representation: 
\begin{align*}
N=\{\x \,|\, \x_1=f(\x')\}
\end{align*}
where we write $\x=(\x_1,\x')$ and $f$ is a $d-1$-variable smooth function. Then we have
\begin{align}\label{surf}
d\s(\x)=\sqrt{1+|\nabla f(\x')|^2}d\x'.
\end{align}

\noindent
\textbf{Acknowledgment.}  
KT was supported by JSPS Research Fellowship for Young Scientists, KAKENHI Grant Number 17J04478 and the program FMSP at the Graduate School of Mathematics Sciences, the University of Tokyo. YT was supported by JSPS Research Fellowship for Young Scientists, KAKENHI Grant Number 17J05051. The authors would like to thank their supervisor Shu Nakamura for encouraging to write this paper and helpful discussions about the Mourre theory for ultrahyperbolic operators.  The authors also would like to thank Evgeny Korotyaev for informing  the paper \cite{KM}.  KT would like to thank Haruya Mizutani for answering many questions about uniform resolvent estimates.

\section{Ultrahyperbolic operators}\label{secul}

Let $d\geq 2$, $0\leq k \leq d$ and let 
\begin{align*}
p(\x)=p_k(\x)=\x_1^2+...+\x_k^2-\x_{k+1}^2-...-\x_d^2
\end{align*}
for $\x\in \re^d$.

\begin{defn}
A differential operator $P$ is called an ultrahyperbolic operator with index $0\leq k \leq d$ if $P$ has the following form:
\begin{align*}
P=\sum_{j=1}^kD_{x_j}^2-\sum_{j=k+1}^dD_{x_j}^2=-\frac{1}{(2\pi)^2}(\sum_{j=1}^k\pa_{x_j}^2-\sum_{j=k+1}^d\pa_{x_j}^2).
\end{align*}
Note that $P=\mathcal{F}^{-1}p\mathcal{F}$.
\end{defn}

In this section, we study resolvent bounds of the ultrahyperbolic operators. Since $h_0$ is a Morse function on $\T^d$, near each critical value $q\in \Gamma$, $h_0$ can be expanded to as the following:
\begin{align*}
h_0(\x)=h_0(q)+(2\pi)^2(\sum_{j=1}^k\y_j^2-\sum_{j=k+1}^d\y_j^2)+O(|\y|^3),
\end{align*}
where $\y=\x-q$. Thus we study the ultrahyperbolic operators for analyzing the resolvent of the discrete Schr\"odinger operator near the thresholds.

\subsection{Limiting absorption principle for ultrahyperbolic operators}

In this subsection, we state a limiting absorption principle for the ultrahyperbolic operators. 
Let $P$ be the ultrahyperbolic operator with index $k$. We define
\begin{align*}
A=\tilde{x}\cdot D_x(I-(2\pi)^{-2}\Delta)^{-1}+(I-(2\pi)^{-2}\Delta)^{-1}D_x\cdot \tilde{x}
\end{align*}
on $C_c^{\infty}(\re^d)$, where $\tilde{x}=(x_1,...,x_k,-x_{k+1},...,-x_d)$. Then, it follows that $P$ and $A$ are essentially self-adjoint on $C_c^{\infty}(\re^d)$ and we also denote the unique self-adjoint extensions by $P$ and $A$ respectively. In fact, for the essential self-adjointness of $P$ it is enough to prove the essential self-adjointness of the multiplication operator $p(\x)$ on $L^2(\re^d)$ by the Fourier transform. However, this is shown since $(p(\x)\pm i)u=0$ and $u\in L^2(\re^d)$ imply $u=0$. For the essential self-adjointness of $A$, we employ Nelson's commutator theorem (see \cite[Theorem X.36]{RS}) with a conjugate operator $-\Delta+|x|^2+1$. 

By a simple calculation, we have
\begin{align*}
[P,iA]=-\pi^{-2}\Delta(I-(2\pi)^{-2}\Delta)^{-1}=\mathcal{F}^{-1}\left( \frac{4|\x|^2}{1+|\x|^2}\mathcal{F} \right).
\end{align*}
In the following, we see that $[P,iA]$ satisfy the Mourre estimate except at $0$.  Note that $E_{I}(P)=\mathcal{F}^{-1}\chi_{I}\circ p\mathcal{F}$, where $E_{I}(P)$ is the spectral projection of $P$ to $I$ and $\chi_I$ is the characteristic function of $I\subset \re$. Fix $I\Subset \re\setminus \{0\}$ and set $a=\inf \{|\l|\,|\, \l\in I\}>0$. Then for $\x\in \supp (\chi_{I}(p(\cdot)))$, we learn
\begin{align*}
|\x|^2=\sum_{j=1}^k|\x_j|^2+\sum_{j=k+1}^d|\x_j|^2\geq a.
\end{align*}
Thus we have
\begin{align*}
\chi_{I}(p(\x))\frac{4|\x|^2}{1+|\x|^2}\chi_{I}(p(\x))\geq \frac{4a}{1+a}\chi_{I}(p(\x))
\end{align*}
and hence
\begin{align*}
E_{I}(P)[P,iA]E_{I}(P)\geq \frac{4a}{1+a}E_{I}(P).
\end{align*}
Moreover since $[P,iA]$ and $[[P,iA],iA]$ are bounded operators, it follows that $P\in C^2(A)$. Thus by the standard Mourre theory (\cite{M}), we have the following proposition.

\begin{prop}\label{lapult}
Let $I\Subset \re\setminus \{0\}$ be a bounded interval and $s>1/2$. Then
\begin{align}\label{la1}
\sup_{z\in I_{\pm}}\|\jap{A}^{-s}(P-z)^{-1}\jap{A}^{-s}\|_{B(L^2(\re^d))}<\infty,
\end{align}
where $I_{\pm}=\{z\in\mathbb{C}\,|\, \Re z\in I,\, \pm\Im z>0\}$. Moreover, the limits
\begin{align}\label{la2}
\jap{A}^{-s}(P-\l\pm i0)^{-1}\jap{A}^{-s}=\lim_{\e\to +0}\jap{A}^{-s}(P-\l\pm i\e)^{-1}\jap{A}^{-s}
\end{align}
exist uniformly in $\l\in I$.
\end{prop}

\begin{rem}\label{laprem}
Since $\jap{A}^s(P-i)^{-1}\jap{x}^{-s}$ is a bounded operator, we can replace $\jap{A}^{-s}$ in $($\ref{la1}$)$ and $($\ref{la2}$)$ by $\jap{x}^{-s}$. 
\end{rem}

\begin{rem}
The Proposition \ref{lapult} is possibly well-known. However, we cannot find a suitable reference and we give a self-contained proof.
\end{rem}

\begin{rem}
By using a scaling argument and Proposition \ref{lapult}, we have a uniform estimate of the high energy limit
\begin{align*}
\sup_{|z|\geq 1}|z|^{1-s}\|\jap{x}^{-s}(P-z)^{-1}\jap{x}^{-s}\|_{B(L^2(\re^d))}<\infty
\end{align*}
for $s>\frac{1}{2}$. 
\end{rem}

\subsection{Uniform resolvent estimates for ultrahyperbolic operators}
In this subsection, we assume $d\geq 3$.

\begin{prop}\label{lapconti}
Let $P$ be an ultrahyperbolic operator. For $\a,\b>\frac{1}{2}+\frac{1}{2(d-1)}$ with $\a+\b\geq 2$, 
\begin{align*}
\sup_{z\in \mathbb{C}\setminus \re}\|\jap{x}^{-\a}(P-z)^{-1}\jap{x}^{-\b}\|_{B(L^2(\re^d))}<\infty.
\end{align*}
\end{prop}

\begin{proof}
This follows from $L^p$-$L^q$ resolvent estimates (see \cite[Theorem 1.1]{JKL}) and a real interpolation argument:
\begin{align*}
\sup_{z\in \mathbb{C}\setminus \re}\|(P-z)^{-1}\|_{B(L^{p,r}(\re^n), L^{q,r}(\re^n))}<\infty
\end{align*}
for 
\begin{align*}
\frac{1}{p}-\frac{1}{q}=\frac{2}{d},\quad \frac{2d(d-1)}{d^2+2d-4}<p<\frac{2(d-1)}{d}, \quad 1\leq r\leq \infty.
\end{align*}
By using H\"older's inequality, we can obtain the following: For $w_1\in L^{\frac{d}{s},\infty}(\re^d)$ and $w_2\in L^{\frac{d}{2-s},\infty}(\re^d)$ with $\frac{1}{2}+\frac{1}{2(d-1)}<s<\frac{3d-4}{2(d-1)}$,
\begin{align}\label{ul1}
\sup_{z\in \mathbb{C}\setminus \re}\|w_1(P-z)^{-1}w_2\|_{B(L^2(\re^d))}\leq C\|w_1\|_{L^{\frac{d}{s},\infty}(\re^d)}\|w_2\|_{L^{\frac{d}{2-s},\infty}(\re^d)}.
\end{align}
In particular, for $\a,\b>\frac{d}{2(d-1)}$ with $\a+\b\geq 2$, 
\begin{align}\label{ul2}
\sup_{z\in \mathbb{C}\setminus \re}\|\jap{x}^{-\a}(P-z)^{-1}\jap{x}^{-\b}\|_{B(L^2(\re^d))}\leq C_{\a\b} <\infty.
\end{align}

\end{proof}

\begin{rem}
In Appendix \ref{appa}, we give a self-contained proof of Proposition \ref{lapconti} with $\a=\b=1$.
\end{rem}

\begin{rem}
Note that if $P$ is elliptic (that is $k=0$ or $k=d$), (\ref{ul1}) holds for $\frac{1}{2}<s<\frac{3}{2}$ and (\ref{ul2}) holds for $\a, \b >\frac{1}{2}$ with $\a+\b\geq 2$ (see \cite[Corollary 2.3]{Mi}). 
\end{rem}

\begin{prop}\label{ulbv}
For $\e>0$, $\jap{x}^{-1-\e}(P-z)^{-1}\jap{x}^{-1-\e}$ is locally H\"older continuous on $B(L^2(\re^d))$ in $\mathbb{C}_{\mp}$. In particular, $\jap{x}^{-1-\e}(P-\l\pm i0)^{-1}\jap{x}^{-1-\e}$ exist in the operator norm topology of $B(L^2(\re^d))$ for $\l\in \re$.
\end{prop}

\begin{proof}

The proof is based on the argument in \cite[Lemma 4.7]{RoS}. Set $L^2_k(\re^d)=\jap{x}^{-k}L^2(\re^d)$ for $k\in \re$.
Note that there are two continuous embeddings $L^{\frac{2d}{d-2}+\d}(\re^d)\subset L_{-1-\e}^{2}(\re^d)$
and $L_{1+\e}^{2}(\re^d)\subset L^{\frac{2d}{d+2}-\d}(\re^d)$ for small $\d>0$. By using $($\ref{disp}$)$ in Appendix \ref{appa}, there exists $\a_{\d}>0$ such that
\begin{align*}
\|&(P-z)^{-1}-(P-z')^{-1}\|_{L_{1+\e}^{2}(\re^d)\to L_{-1-\e}^{2}(\re^d)}\\
&\leq C\|(P-z)^{-1}-(P-z')^{-1}\|_{B(L^{\frac{2d}{d+2}-\d}(\re^d), L^{\frac{2d}{d-2}+\d}(\re^d))}\\
&=C\left\|\int_{0}^{\infty}(e^{itz}-e^{itz'})e^{-itP}dt\right\|_{B(L^{\frac{2d}{d+2}-\d}(\re^d), L^{\frac{2d}{d-2}+\d}(\re^d))}\\
&\leq C\int_{0}^{\infty}\min{\left(2,|z-z'|t\right)}\frac{1}{t^{1+\a_{\d}}}dt\\
&=C\int_0^{|z-z'|^{-1}}\frac{|z-z'|}{t^{1+\a_{\d}}}dt+C\int_{|z-z'|^{-1}}^{\infty}\frac{1}{t^{1+\a_{\d}}}dt\\
&\leq C|z-z'|^{\a_{\d}}.
\end{align*}
The existence of the boundary values $\jap{x}^{-1-\e}(P-\l\pm i0)^{-1}\jap{x}^{-1-\e}$ directly follows from the H\"older continuity.
\end{proof}

Next, we state the optimality of the estimate. For a preparation, we need the following lemma.
\begin{lem}\label{reg}
Let $d\geq 3$, $r>0$ and $\f(x)=\chi(x)|x|^{-\frac{d-2}{2}}$, where $\chi\in C_c^{\infty}(\re^d)$ with $\chi=1$ on $|x|\leq r$. Then, $\f\in H^{s}(\re^d)$ for $0\leq s<1$ and $\f\notin H^1(\re^d)$.
\end{lem}

\begin{proof}
Note that $\f\in L^2(\re^d)$. We learn
\begin{align*}
\pa_{x_j}\f(x)=(\pa_{x_j}\chi(x))|x|^{-\frac{d-2}{2}}-\frac{d-2}{2}x_j|x|^{-\frac{d+2}{2}},
\end{align*}
and hence
\begin{align*}
|\nabla\f(x)|\geq C|x|^{-\frac{d}{2}}
\end{align*}
near $x=0$. Thus, $|\nabla\f|\notin L^2(\re^d)$ and hence $\f\notin H^1(\re^d)$. 

Next, we show that $\f \in H^s(\re^d)$ for $0\leq s<1$.  It suffices to prove that $\jap{\x}^s\hat{\f}(\x)\in L^2(\re^d)$ where $\hat{\f}(\x)=\mathcal{F}\f(\x)$. Note that $\f\in C^{\infty}(\re^d \setminus \{0\})$ and $\widehat{|x|^{-\frac{d-2}{2}}}(\x)=c_d|\x|^{-\frac{d}{2}-1}$, where $c_d$ is a constant depending only on $d$.
We learn
\begin{align*}
\hat{\f}(\x)=c_d \int_{\re^d}\hat{\chi}(\y)|\x-\y|^{-\frac{d}{2}-1}d\y.
\end{align*}
Since 
\begin{align*}
\left|\int_{|\x-\y|\leq\frac{1}{2}|\x|}\hat{\chi}(\y)|\x-\y|^{-\frac{d}{2}-1}d\y\right|\leq& \left|\int_{\frac{1}{2}|\x|\leq |\y|\leq \frac{3}{2}|\x|}\hat{\chi}(\y)|\x-\y|^{-\frac{d}{2}-1}d\y\right|\\
\leq&C\jap{\x}^{-N}\int_{\frac{1}{2}|\x|\leq |\y|\leq \frac{3}{2}|\x|}|\x-\y|^{-\frac{d}{2}-1}d\y\\
\leq&C\jap{\x}^{-N-1+\frac{d}{2}},
\end{align*}
and
\begin{align*}
\left|\int_{|\x-\y|>\frac{1}{2}|\x|}\hat{\chi}(\y)|\x-\y|^{-\frac{d}{2}-1}d\y\right|\leq& C|\x|^{-\frac{d}{2}-1}\int_{|\x-\y|>\frac{1}{2}|\x|}\hat{\chi}(\y)d\y\\
\leq&C|\x|^{-\frac{d}{2}-1}
\end{align*}
for any $|\x|\geq 1$ and any $N>0$, we have $\jap{\x}^s\hat{\f}\in L^2(\{|\x|\geq1\})$. This and $\f\in L^2(\re^d)$ imply $\jap{\x}^s\hat{\f}\in L^2(\re^d)$.
\end{proof}

Using the above lemma, we obtain:

\begin{prop}\label{ulcoex}
For $0\leq s<1$, we have
\begin{align*}
\sup_{z\in \mathbb{C}\setminus \s(P)}\|\jap{x}^{-s}(P-z)^{-1}\jap{x}^{-s}\|_{B(L^2(\re^d))}=\infty.
\end{align*}
\end{prop}

\begin{proof}
For simplicity, we deal with $k=1$ only. We may assume $s>1/2$. Note that by Proposition \ref{lapult} and Remark \ref{laprem}, $\jap{x}^{-s}(P+\e\pm i0)^{-1}\jap{x}^{-s}$ exist in $B(L^2(\re^d))$ for $\e\neq 0$. Moreover, it follows that
\begin{align*}
&\jap{x}^{-s}(P+\e- i0)^{-1}\jap{x}^{-s}-\jap{x}^{-s}(P+\e+ i0)^{-1}\jap{x}^{-s}\\
&=\jap{x}^{-s}\delta(P+\e)\jap{x}^{-s}
\end{align*}
due to Stone's theorem. Thus it suffices to prove that
\begin{align*}
\sup_{\e>0}\|\jap{x}^{-s}\delta(P+\e)\jap{x}^{-s}\|_{B(L^2(\re^d))}=\infty.
\end{align*}
By using the Fourier transform, it is sufficient to find $\f\in H^s(\re^d)$ such that
\begin{align*}
\sup_{\e>0}\left|(\f,\delta(p+\e)\f)_{L^2(\re^d)}\right|=\infty.
\end{align*}
Note that
\begin{align*}
(\f,\delta(p(\x)+\e)\f)=\int_{p(\x)=-\e}|\f(\x)|^2\frac{d\s(\x)}{|\nabla_{\x}p|}=\int_{p(\x)=-\e}|\f(\x)|^2\frac{d\s(\x)}{2|\x|},
\end{align*}
and
\begin{align*}
p(\x)=-\e\Leftrightarrow \x_1^2=\sum_{j=2}^d\x_j^2+\e.
\end{align*}
Using the formula $($\ref{surf}$)$, we learn
\begin{align*}
(\f,\delta(p(\x)+\e)\f)=&\sum_{\pm}\int_{\re^{d-1}}|\f(\x)|^2\sqrt{1+\left|\frac{\x'}{\sqrt{|\x'|^2+\e}}\right|^2}\frac{d\x'}{2|\x|}\\
=&\sum_{\pm}\int_{\re^{d-1}}|\f(\x)|^2\frac{d\x'}{2\sqrt{|\x'|^2+\e}},
\end{align*}
where $\x'=(\x_2,...,\x_d)$ and $\x=(\pm \sqrt{|\x'|^2+\e}, \x')$ for $\x\in \{p(\x)=\e\}$. Thus we now take $\f(\x)=\frac{1}{|\x|^{(d-2)/2}}\chi(\x)$, where $\chi\in C_c^{\infty}(\re^d)$ such that $\chi=1$ on $|\x|\leq 1$. Note that $\f\in H^s(\re^n)$ due to Lemma \ref{reg}. Since 
\begin{align*}
\int_{ \x'\in\re^{d-1}, |\x'|\leq 1}\frac{1}{|\x'|^{d-1}}d\x'=\infty,
\end{align*}
$\f$ has the desired property.
\end{proof}

\begin{rem}\label{sculrem}
This proposition also follows from a scaling argument. In fact, for $\a,\b>\frac{1}{2}$ we have
\begin{align*}
&\|(1+|x|)^{-\a}(P-z)^{-1}(1+|x|)^{-\b}\|_{B(L^2(\re^d))}\\
=&\e^{2-(\a+\b)}\|(\e^{-1}+|x|)^{-\a}(P-\e^2z)^{-1}(\e^{-1}+|x|)^{-\b}\|_{B(L^2(\re^d))}\\
\leq&\e^{2-(\a+\b)}\|(1+|x|)^{-\a}(P-\e^2z)^{-1}(1+|x|)^{-\b}\|_{B(L^2(\re^d))}.
\end{align*}
for $0<\e<1$. If we take supremum of $z\in \mathbb{C}\setminus \re$ and take $\e\to 0$, then we obtain  a contradiction unless $\a+\b\geq 2$. For the Laplace operators, see \cite{BR}. However we give a more direct proof for a special case since the above argument can be applicable to the discrete Schr\"odinger operators near the hyperbolic thresholds. See Remark \ref{remcoex}. 
\end{rem}

\section{Proofs of positive results}\label{secpo}

In this section, we prove Theorems \ref{mainpo}, \ref{bv} and Corollary \ref{cor}.
\subsection{Proof of Theorem \ref{mainpo} $(i)$}\label{sttouni}

It is known that there is a deep connection between the time decay of the Schr\"odinger propagator $e^{itP}$ and the threshold property of the resolvent of $P$. We refer \cite[\S XIII-A]{RS}. Here we employ a bit technical, but very strong tool due to Duyckaerts. His method allows us to deduce $L^p$-$L^q$ uniform resolvent estimates from Strichartz estimates.

First, we state the dispersive estimates and the Strichartz estimates for the discrete Schr\"odinger operators.

\begin{prop}[\cite{SK}]
Let $d\geq 1$. Then, there exists $C>0$ such that for any $t\in \re$
\begin{align*}
\|e^{-itH_0}\|_{l^1(\ze^d)\to l^{\infty}(\ze^d)}\leq C\jap{t}^{-\frac{d}{3}}.
\end{align*}
\end{prop}

\begin{cor}
Let $d\geq 4$. Set $3^*=\frac{2d}{d-3}$ and $3_*=\frac{2d}{d+3}$.
Then, we have the following Strichartz estimates: 
Suppose that $u\in C(\re, l^2(\ze^d))$ and $F\in L^2(\re, l^{3_*,2}(\ze^d))$ satisfy
\begin{align}\label{loSt}
i\pa_tu(t)-H_0u(t)=F,\quad u|_{t=0}=f\in l^2(\ze^d).
\end{align}
Then there exists $C>0$ such that for $0<T\leq \infty$ we have
\begin{align*}
\|u\|_{L^{2}(-T,T)l^{3^*,2}(\ze^d)}\leq C\|f\|_{l^2(\ze^d)}+C\|F\|_{L^{2}(-T,T)l^{3_{*},2}(\ze^d)}. 
\end{align*}

\end{cor}
\begin{proof}
We apply Theorem 10.1 in \cite{KT} with $\mathcal{H}=B_0=l^{2}(\ze^d)$, $B_1=l^1(\ze^d)$, $\s=\frac{d}{3}$ and $q=2$. 
\end{proof}

The next argument is due to T. Duyckaerts (see \cite[Proposition 5.1]{BM}).

\begin{prop}\label{discuni}
Let $d\geq 4$. Then, there exists $C>0$ such that for $z\in \mathbb{C}\setminus \s(H_0)$
\begin{align*}
\|(H_0-z)^{-1}u\|_{l^{3^*,2}(\ze^d)}\leq C\|u\|_{l^{3_*,2}(\ze^d)},\quad u\in l^2(\mathbb{Z}^d)\cap l^{3_*,2}(\mathbb{Z}^d).
\end{align*}
Moreover, for $w\in l^{\frac{2d}{3},\infty}(\re^d)$,
\begin{align*}
\sup_{z\in\mathbb{C}\setminus \re}\|w(H_0-z)^{-1}w\|_{B(\mathcal{H})}\leq C\|w\|_{l^{\frac{2d}{3},\infty}(\ze^d)}^2.
\end{align*}
In particular, if $V\in l^{\frac{d}{3},\infty}(\ze^d)$, (\ref{BS}) holds.
\end{prop}

\begin{proof}
Suppose that $f$ is a finitely supported function on $\ze^d$. Let $z\in \mathbb{C}\setminus \s(H_0)$. We substitute $u(t)=e^{itz}f$ into (\ref{loSt}) and then we have
\begin{align*}
\gamma(z,T)\|f\|_{l^{3^*,2}(\ze^d)}\leq C\|f\|_{l^2(\ze^d)}+C\gamma(z,T)\|(H_0-z)f\|_{l^{3_*,2}(\ze^d)},
\end{align*}
where $\gamma(z,T)=\|e^{izt}\|_{L^2(-T,T)}$. Since $\gamma(z,T)\geq \sqrt{T}$, by letting $T\to \infty$,
\begin{align*}
\|f\|_{l^{3^*,2}(\ze^d)}\leq C\|(H_0-z)f\|_{l^{3_*,2}(\ze^d)}.
\end{align*}
It remains to justify a density argument. 
\end{proof}

\subsection{Proof of Theorem \ref{mainpo} $(ii)$}\label{disL^2}

In this subsection, we assume $d\geq 3$. 
\begin{prop}\label{discuni2}
For $\a,\b>\frac{1}{2}+\frac{1}{2(d-1)}$ with $\a+\b\geq 2$, there exists $C>0$ such that 
\begin{align*}
\sup_{z\in \mathbb{C}\setminus \re}\|\jap{x}^{-\a}(H_0-z)^{-1}\jap{x}^{-\b}\|_{B(\mathcal{H})}\leq C.
\end{align*}
\end{prop}

\begin{proof}

By using a partition of unity, it suffices to prove for each $\chi\in C^{\infty}(\T^d)$ with a small support, $f \in H^{\a}(\mathbb{T}^d)$ and  $g \in H^{\b}(\mathbb{T}^d)$
\begin{align}\label{lap1}
|(f,\chi^2(h_0-z)^{-1}g)_{L^2(\mathbb{T}^d)}|\leq C\|f\|_{H^{\a}(\mathbb{T}^d)}\|g\|_{H^{\b}(\mathbb{T}^d)},
\end{align}
where $C>0$ is independent of $f$, $g$ and $z$.
We may suppose $\chi$ has one of the following properties: $\nabla h_0\neq 0$ on $\supp\chi$ or $\supp\chi$ contains just one element of $\Gamma$. Since $($\ref{lap1}$)$ follows from Proposition \ref{cutMo} in the former case, we may only deal with the latter case. We take a unique element $\x_0\in \Gamma\cap \supp \chi$. Then there exists a diffeomorphism $\k$ from a neighborhood of $\supp \chi$ onto its image such that
\begin{align*}
h_0(\k^{-1}(\y))=&h_0(\x_0)+\y_1^2+...+\y_{k}-\y_{k+1}^2-...-\y_d^2,\quad \y \in \k(\supp \chi)\subset \re^d
\end{align*}
for some $0\leq k\leq d$. Set $J(\y)=|\det d\k^{-1}(\y)|$ and $f_{\k}(\y)=f(\k^{-1}(\y))$. By using the change of variables and Proposition \ref{lapconti}, we have
\begin{align*}
|(f,\chi^2(h_0-z)^{-1}g)_{L^2(\mathbb{T}^d)}|=&\left|\int_{\mathbb{T}^d}\frac{\chi(\x)^2\bar{f}(\x)g(\x)}{h_0(\x)-z}d\x\right|\\
=&\left|\int_{\re^d}\frac{\chi_{\k}(\y)^2\bar{f}_{\k}(\y)g_{\k}(\y)}{p_k(\y)+h_0(\x_0)-z}J(\y)d\y\right|\\
\leq& C\|\chi_{\k}f_{\k}\|_{H^{\a}(\re^d)}\|\chi_{\k}g_{\k}\|_{H^{\b}(\re^d)}\\
\leq& C\|f \|_{H^{\a}(\mathbb{T}^d)}\|g\|_{H^{\b}(\mathbb{T}^d)},
\end{align*}
where we used Lemma \ref{lemb3} in the last inequality. Thus we obtain $($\ref{lap1}$)$.

 \end{proof}

\subsection{Proof of Theorem \ref{bv}}\label{pfbv}
Let $s>1$ and $\a_{\d}>0$ as in the proof of Proposition \ref{ulbv}. Similarly to Subsection \ref{disL^2}, it is enough to prove that
\begin{align*}
|(f,\chi^2((h_0-z)^{-1}-(h_0-z')^{-1})g)|\leq C|z-z'|^{\a_{\d}}\|f\|_{H^s(\mathbb{T}^d)}\|g\|_{H^s(\mathbb{T}^d)}
\end{align*}
for $\chi\in C^{\infty}(\T^d)$ which is as in subsection \ref{disL^2}. However, it is proved by changing variables, Proposition \ref{ulbv}, Lemma \ref{lemb2} and Proposition \ref{cutMo}.

\subsection{Proof of Corollary \ref{cor}}\label{pfcor}

Corollary \ref{cor} follows from Lemma \ref{numberEV}.
The argument in the proof is due to \cite[Lemma 2.1]{HMO}.

\begin{lem}\label{numberEV}
Let $H$ be a bounded self-adjoint operator on a Hilbert space $\mathcal{H}$ which has no eigenvalues and $W$ be a finite rank self-adjoint operator on $\mathcal{H}$. Then:
\begin{itemize}
\item[(i)] For any $\m \in \re$, 
$\dim \Ker (H+W-\m)  \leq \dim \Ran W$.

\item[(ii)] Suppose that $\sigma(H)=[a,b]$, $-\infty<a<b<\infty$ and $W=W_+ - W_-$, where $W_\pm$ are positive operators.
Then 
\begin{align*}
&\dim(\Ran E_{H+W}^{pp}((-\infty,a ])) \leq \dim\operatorname{Ran} W_-,\\
&\dim(\Ran E_{H+W}^{pp}([ b,\infty))) \leq \dim\operatorname{Ran} W_+. 
\end{align*}
\end{itemize}
\end{lem}

\begin{proof}
(i) 
Suppose that the inequality fails. Let $P$ be the projection onto $\Ran W = (\Ker W)^{\perp}$. Then $P\mid_{\Ker(H+W-\m)}: \Ker(H+W-\m) \to \Ran W$ has a non-trivial kernel, i.e.\ we can choose $u \in \Ker(H+W-\m)$ such that $Wu=0$ and $\|u\|_{\H}=1$. Therefore
\begin{align*}
0=(H+W-\m)u=(H-\m)u,
\end{align*}
which contradicts the assumption that $H$ has no eigenvalues.


(ii) Suppose that the first inequality fails. Then the same argument as in (i) implies that there exists $u \in \Ran E_{H+W}^{pp}((-\infty,a])$ such that $W_- u=0$ and $\|u\|_{\H}=1$. Therefore we have
\begin{align*}
a \geq (u,(H+W)u)_{\H} = (u,(H+W_+)u)_{\H} \geq (u,Hu)_{\H},
\end{align*}
where the last inequality follows from the positivity of $W_+$.
On the other hand, the assumption on $H$ implies $(u,Hu)_{\H} \in (a,b)$, which is a contradiction.
The other inequality is similarly proved.
\end{proof}


\section{Proofs of negative results}\label{secne}

\subsection{Proof of Theorem \ref{mainne} $(i)$}

The following argument is similar to \cite[Theorem XIII.11]{RS}. Note that $h_0(\x)\sim 4\pi^2 |\x|^2$ near $\x=0$ and the operator $H_0$ is positive.

Set $K_{\m}=|V|^{1/2}(H_0+\m^2)^{-1}|V|^{1/2}$ for $\m\in \re$. First, we note that $H_0+\l V$ has a negative eigenvalue if and only if there exists $\m>0$ such that $1/\l$ is an eigenvalue of $K_\m$. In fact, a direct calculus implies
\begin{align*}
H_0 u+\l Vu=-\m^2u\Leftrightarrow& \l (H_0 +\m^2)^{-1}Vu=-u\\
\Rightarrow& \l K_{\m}\g=\g,
\end{align*}
where $\g=|V|^{\frac{1}{2}}u$.
Conversely, if  there exists $\g\in \mathcal{H}$ such that $\l K_{\m}\g=\g$, then
\begin{align*}
\l |V|u=|V|^{1/2}\g=(H_0 +\m^2)u,
\end{align*}
where $u=(H_0+\m^2)^{-1}|V|^{1/2}\g$.

Since $V$ vanishes at infinity, $K_{\m}$ is a positive compact operator. Then it suffices to prove that
\begin{align*}
\lim_{\m^2\to 0}\|K_{\m}\|_{B(\H)}=\infty.
\end{align*}
In fact, the spectral radius of $K_{\m}$ is equal to $\|K_{\m}\|_{B(\H)}$, $\s(K_{\m})=\s_{pp}(K_{\m})$ and  $\s(K_{\m})$ has no accumulation point except at $0$. Thus we need only find $\y\in \H$ such that
\begin{align*}
\lim_{\m^2\to 0}\left(|V|^{1/2}\y, (H_0+\m^2)^{-1}|V|^{1/2}\y\right)_{\H}=\infty.
\end{align*}
We choose a non negative finitely supported function $\y \in \mathcal{H}$ which satisfies $\y(x)> 0$ for some $x\in \supp V$ and set $\f=|V|^{1/2}\y$. Then $\hat{\f}(0)=\sum_{x\in\ze^d}|V(x)|^{1/2}\y(x)>0$. Since $\f$ is finitely supported, then $\hat{\f}\in L^2(\T^d)\cap C(\T^d)$. Thus, $\hat{\f}\neq 0$ near zero. Note that $h_0(\x)\sim 4\pi^2 |\x|^2$ near $\x=0$. Consequently, 

\begin{align*}
\left(|V|^{1/2}\y, (H_0+\m^2)^{-1}|V|^{1/2}\y\right)_{\H}=\int_{\mathbb{T}^d}\frac{|\hat{\f}(\x)|^2}{h_0(\x)+\m^2}d\x
\end{align*}
diverges as $\m^2\to 0$ if $d=1$ or $2$.

\subsection{Proof of Theorem \ref{mainne} $(ii)$}
We consider near $\x_1=\x_2=\frac{1}{4}$ only, the other cases being similar. 
\begin{lem}\label{2flat}
In a neighborhood of $(\frac{1}{4},\frac{1}{4})\in \mathbb{T}^2$, $h_0(\x)=4$ is equivalent to $\x_1+\x_2=\frac{1}{2}$.
\end{lem}

\begin{proof}
Note that $h_0(\x)=4\sin^2\pi \x_1+4\sin^2\pi \x_2=4-2\cos 2\pi\x_1-2\cos 2\pi \x_2$. Thus, $h_0(\x)=4$ is equivalent to 
\begin{align*}
\cos 2\pi\x_1+\cos 2\pi \x_2=\cos(\pi(\x_1+\x_2))\cos(\pi(\x_1-\x_2))=0.
\end{align*}
Near $\x_1=\x_2=\frac{1}{4}$, this is equivalent to $\x_1+\x_2=\frac{1}{2}$.
\end{proof}

\begin{prop}\label{prop2d}
Let $\chi \in C^{\infty}(\mathbb{T}^2)$ be a non-negative function which is equal to $1$ near $\x_1=\x_2=1/4$ and  is supported near $\x_1=\x_2=1/4$.
If $q\neq \infty$,
\begin{align*}
(H_0-4\pm i0)^{-1}\mathcal{F}_d^{-1}(\chi)\notin l^q(\ze^2).
\end{align*}
\end{prop}
\begin{proof}
Let us denote $\mathcal{H}_s=\jap{x}^{-s}\mathcal{H}$.
If we take the support of $\chi$ small near $\x_1=\x_2=\frac{1}{4}$, then $\chi(D)(H_0-4\pm i0)^{-1}$ exists in $B(\H_s, \H_{-s})$ for $s>\frac{1}{2}$ since $\nabla h_0(\x)\neq 0$ on $\supp \chi$. Here we used Lemma \ref{cutMo}. Then, it suffices to prove that $\left((H_0-4- i0)^{-1}-(H_0-4+ i0)^{-1}\right)\mathcal{F}_d^{-1}\chi\notin l^q(\ze^2)$ for $q\neq \infty$. Stone's theorem implies
\begin{align*}
&\frac{1}{2\pi i}((H_0-4- i0)^{-1}-(H_0-4+ i0)^{-1})\mathcal{F}_d^{-1}\chi\\
&=\d(H_0-4)\mathcal{F}_d^{-1}\chi\\
&=\int_{h_0(\x)=4}e^{2\pi ix_1\x_1+2\pi ix_2\x_2}\chi(\x)\frac{d\s(\x)}{|\nabla h_0(\x)|}.
\end{align*}
By using Lemma \ref{2flat} and the formula $($\ref{surf}$)$,
\begin{align*}
I(x_1,x_2):=&\int_{h_0(\x)=4}e^{2\pi ix_1\x_1+2\pi ix_2\x_2}\chi(\x)\frac{d\s(\x)}{|\nabla h_0(\x)|}\\
=&\int_{\re}e^{2\pi ix_1\x_1+2\pi ix_2(\frac{1}{2}-\x_1)}\chi(\x_1,\frac{1}{2}-\x_1)\frac{d\x_1}{4\pi \sin(2\pi \x_1)}
\end{align*}
is rapidly decreasing with respect to $|x_1-x_2|$ since $\chi(1/4,1/4)=1$. However, we cannot obtain any decay with respect to $|x_1+x_2|$. We write
\begin{align*}
I(x_1,x_2)=e^{\pi ix_2}J(x_1-x_2),\,\, J(t)=\int_{\re}e^{2\pi it\x_1}\chi(\x_1,\frac{1}{2}-\x_1)\frac{d\x_1}{4\pi \sin(2\pi \x_1)}.
\end{align*}
We employ the change of variables: $s=x_1+x_2$ and $t=x_1-x_2$ and write $I_1(s,t)=I(x_1,x_2)$.
Since $|I_1(s,t)|=|J(t)|$ is independent of $s$, $|I_1(s,t)|\nrightarrow 0$ as $|s|\to \infty$ unless $t\in \{J(t)= 0\}$. By the assumption of $\chi$, we have $|I_1(s,0)|=|J(0)|\neq 0$. Thus $|I|$ does not decay with respect to $s=x_1+x_2$.
As a consequence, $((H_0-4- i0)^{-1}-(H_0-4+i0)^{-1})\mathcal{F}^{-1}_d(\chi)$ does not belong to $l^q(\ze^2)$ unless $q=\infty$.
\end{proof}
The above proposition shows that if $d=2$, $\chi(D)(H_0-z)^{-1}$ is not bounded from $l^p(\ze^2)$ to $l^q(\ze^2)$ unless $q=\infty$. Thus, the proof of the second part of Theorem \ref{mainne} $($ii$)$ is completed.


In the rest of the subsection, we prove the first part of Theorem \ref{mainne} $($ii$)$.
Let $w(x_1,x_2)=\jap{x_1+x_2}^{-1/2}\jap{x_1-x_2}^{-1}\in l^{2,\infty}(\ze^2)$ and $\g$ be a non zero finitely supported function such that $\g\geq 0$. Let $u(x)=e^{-\pi i\frac{x_1+x_2}{2}}(w^{-1}\g)(x)$. Note that
\begin{align*}
\mathcal{F}_d^{-1}(wu)(1/4,1/4)=\sum_{x\in \ze^2}e^{2\pi i\times \frac{x_1+x_2}{4}}e^{-\pi i\frac{x_1+x_2}{2}}\g(x)=\sum_{x\in \ze^2}\g(x)>0.
\end{align*}
Thus $|w(x)\chi(D)\delta(H_0-4)wu(x)|\sim C (1+|x_1+x_2|)^{-1/2}(1+|x_1-x_2|)^{-\infty}$ as in the proof of Proposition \ref{prop2d} and the right hand side does not belong to $\H$.

\subsection{Proof of Theorem \ref{mainne} $(iii)$}

In Proposition \ref{discuni} and Theorem \ref{discuni2}, we have seen uniform bounds of Birman-Schwinger operator for $V\in l^{d/3,\infty}(\ze^d)$ or $V(x)=\jap{x}^{-2}$. Since $\jap{x}^{-2}\in l^{d/2,\infty}(\ze^d)$, it is natural to ask whether it is true for general potentials $V\in l^{d/2,\infty}(\ze^d)$. However, the next proposition says that it is false at least if $d\geq 5$. 

\begin{prop}
Let $d\geq 3$ and 
\begin{align*}
w(x)=w_p(x)=\jap{x_d}^{-1/p}\prod_{j=1}^{d-1}\jap{x_j-x_d}^{-1/p}\in l^{p,\infty}(\ze^d),\quad p>0.
\end{align*}
Suppose that
\begin{align}\label{colap}
\sup_{\l\in \re\setminus \s(H_0)}\|w_p(H_0-\l)^{-1}w_p\|_{B(\H)}<\infty.
\end{align}
Then, $p\leq 2(d+2)/3$ holds. In particular, if $d\geq 5$ then $w_d$ does not satisfy $($\ref{colap}$)$.
\end{prop}

\begin{proof}
We construct a variant of the Knapp counter example near the energy surface $h_0(\x)=2d$. We denote the $d$-dimensional Fourier expansion $\mathcal{F}_d$ of $u$ by $\hat{u}$ and the one dimensional Fourier transform by $\mathcal{F}^1$.  We take a real valued function $\chi\in C_c^{\infty}\left((-\frac{1}{4}, \frac{1}{4})\right)$ such that $\chi=1$ near $0$. We can regard $\chi$ as a function on $S^1$. Let
\begin{align*}
\f_{\e}(x)=&e^{2\pi i(dx_d/4+\sum_{j=1}^{d-1}x_j/4)}a\e^{d+2}w_{p}^{-1}(x)((\mathcal{F}^1)^{-1}{\chi})(a\e^3x_d)\\
&\times\prod_{j=1}^{d-1}((\mathcal{F}^1)^{-1}\chi)(\e (x_j-x_d))
\end{align*}
for $0<\e\leq 1$, $a>0$ and $x\in \ze^d$. Then,
\begin{align*}
\widehat{w\f_{\e}}(\x)=\chi \left(\sum_{j=1}^d{\frac{\x_j-1/4}{a\e^3}}\right) \prod_{j=1}^{d-1}\chi\left(\frac{\x_j-1/4}{\e}\right)\in C^{\infty}(\mathbb{T}^d),
\end{align*}
where $\x\in [0,1)^d$ and we regard the function $\widehat{w\f_{\e}}$ on $[0,1)^d$ as a function on $\T^d$ by virtue of the support property of $\chi$. 
Note that $w\f_{\e}$ is rapidly decreasing and $\widehat{w\f_{\e}}$ has a small support near $\{\x_j=1/4,\,j=1,...,d\}$ which does not contain critical points of $h_0(\x)=4\sum_{j=1}^d\sin^2 (\pi\x_j)$. Thus $(\f_{\e},w\delta(H_0-2d)w\f_{\e})_{\mathcal{H}\to \mathcal{H}}$ exists by Proposition \ref{cutMo}. We observe that if $\x'=(\x_1,...,\x_{d-1}) \in \supp \left(\prod_{j=1}^{d-1}\chi(\frac{\x_j-1/4}{\e})\right)$ and $\x\in h_0^{-1}(\{2d\})$, then
\begin{align}\label{order}
\sum_{j=1}^d(\x_j-1/4)=O(\e^3).
\end{align}
In fact, by using the Taylor expansion near $\{\x_j=1/4,\,j=1,...,d\}$, we have
\begin{align*}
0=h_0(\x)-2d=&4\sum_{j=1}^d(\x_j-1/4)+O\left(\sum_{j=1}^d(\x_j-1/4)^3\right)\\
=&4\sum_{j=1}^d(\x_j-1/4)+O(\e^3)+O((\x_d-1/4)^3).
\end{align*}
This implies (\ref{order}). Therefore, if we take $a>0$ large enough (which remains to be independent of $\e$), it follows that
\begin{align*}
\supp \left(\prod_{j=1}^{d-1}\chi\left(\frac{\x_j-1/4}{\e}\right)\right)\cap h_0^{-1}(\{2d\})\subset \supp\chi\left(\sum_{j=1}^d\frac{\x_j-1/4}{a\e^3}\right).
\end{align*}
By using this, we obtain

\begin{align*}
(\f_{\e},w\delta(H_0-2d)w\f_{\e})_{\H}=&(\widehat{w\f_{\e}}, \delta(h_0-2d)\widehat{w\f_{\e}})_{L^2(\mathbb{T}^d)}\\
=&\int_{h_0(\{2d\})\cap (-\frac{1-\e}{4},\frac{1+\e}{4})^d}|\widehat{w\f_{\e}}(\x)|^2 \frac{d\s(\x)}{|\nabla{h_0(\x)}|}\\
\geq& C\e^{d-1}
\end{align*}
for some $C>0$ which is independent of $\e$.

On the other hand, we observe that for $s>2$
\begin{align*}
\sum_{x_j\in \mathbb{Z}}&\jap{x_j-x_d}^{2/p}|((\mathcal{F}_d^1)^{-1}\chi)(\e(x_j-x_d))|^2\\
\leq& C\sum_{x_j\in \mathbb{Z}}\jap{x_j-x_d}^{2/p}\jap{\e(x_j-x_d)}^{-2s}\\
=&C\left(\sum_{|x_j|<1/\e}+\sum_{|x_j|\geq 1/\e}\right)\jap{x_j}^{2/p}\jap{\e x_j}^{-2s}\leq C\e^{-1-2/p}.
\end{align*}
Then, we obtain
\begin{align*}
\|\f_{\e}\|_{\H}^2\leq C\e^{2(d+2)}\cdot \e^{(d+2)(-1-2/p)}=C\e^{(d+2)(1-2/p)}.
\end{align*}
By using (\ref{colap}), we have $\e^{d-1}\leq C\e^{(d+2)(1-2/p)}$. Since this holds for small $\e>0$, we conclude $p\leq 2(d+2)/3$.

\end{proof}

\subsection{Proof of Theorem \ref{mainne} $(iv)$}

For $0\leq s<1$, we prove
\begin{align*}
\sup_{z\in \mathbb{C}\setminus \s(H_0)}\|\jap{x}^{-s}(H_0-z)^{-1}\jap{x}^{-s}\|_{B(\H)}=\infty.
\end{align*}
It suffices to prove that there exists $\f\in H^s(\mathbb{T}^d)$ such that
\begin{align*}
\lim_{\e\to 0, \e>0}|(\f,(h_0(\x)+\e)^{-1}\f)_{L^2(\mathbb{T}^d)}|=\infty.
\end{align*}
Fix $\x_0\in h_0^{-1}(\{0\})$. Then there exists a diffeomorphism $f$ from a small neighborhood of $\x_0$ to a small open ball in $\re^d$ such that $h_0(f^{-1}(\y))=|\y|^2$. We take
\begin{align*}
\f(\x)=\chi(\x)f^{*}\left(\frac{1}{|\y|^{\frac{d-2}{2}}}\right)(\x),
\end{align*}
where $\chi\in C^{\infty}(\mathbb{T}^d)$ has a small support near $\x_0$. Note that $\f\in H^{s}(\mathbb{T}^d)$ for $0\leq s<1$ due to Lemma \ref{reg}. Thus, we obtain
\begin{align*}
|(\f,(h_0(\x)+\e)^{-1}\f)_{l^2}|\geq C\int_{\y\in \re^d, |\y|\text{:near}\,0}\frac{1}{|\y|^{d-2}(|\y|^2+\e)}d\y\to \infty
\end{align*}
as $\e\to 0$.

\begin{rem}\label{remcoex}
In the above proof, we have constructed a function supported near an elliptic threshold. However, this argument is applicable to near a hyperbolic threshold. See the proof of Proposition \ref{ulcoex}. 
\end{rem}

\appendix

\section{Self-contained proof of Proposition \ref{lapconti} in a particular case}\label{appa}
We can apply the argument in Subsection \ref{sttouni}  to the ultrahyperbolic operators $P$: Note that
\begin{align*}
e^{4\pi^2 itP}u(x,y)=&e^{-it\Delta_x}e^{it\Delta_y}u(x,y)\\
=&\frac{1}{(-4\pi i t)^{\frac{k}{2}}}\frac{1}{(4\pi i t)^{\frac{d-k}{2}}}\int_{\re^k_{x'}}\int_{\re^{d-k}_{y'}}e^{\frac{|x-x'|^2}{4it}}e^{-\frac{|y-y'|^2}{4it}}u(x',y')dy'dy',
\end{align*}
where $x\in \re^k$ and $y\in \re^{d-k}$.
Thus, we obtain the following dispersive estimates: 
\begin{align*}
\|e^{4\pi^2 itP}\|_{L^1(\re^d)\to L^{\infty}(\re^d)}\leq \frac{1}{(4\pi|t|)^{\frac{d}{2}}},\quad t\in \re\setminus \{0\}.\\
\end{align*}
Using a complex interpolation, we have
\begin{align}\label{disp}
\|e^{4\pi^2itP}\|_{L^p(\re^d)\to L^{p'}(\re^d)}\leq \frac{1}{(4\pi|t|)^{\frac{d}{2}(\frac{1}{p}-\frac{1}{p'})}}
\end{align}
for $1\leq p\leq 2$ and $p'=(p-1)/p$.
By using the unitarity of $e^{4\pi^2itP}$ and \cite[Theorem 10.1]{KT}, we have the following:

Let $d\geq 3$ and $P$ be an ultrahyperbolic operator. Let $2^*=\frac{2d}{d-2}$ and $2_*=\frac{2d}{d+2}$. Suppose that $u\in C(\re, L^2(\re^d))$ and $F\in L^2L^{2_*,2}$ satisfy
\begin{align}\label{loStul}
i\pa_tu(t)-Pu(t)=F,\quad u|_{t=0}=f\in L^2(\re^d).
\end{align}
Then there exists $C>0$ such that for $0<T\leq \infty$ we have
\begin{align*}
\|u\|_{L^{2}(-T,T)L^{2^*,2}(\re^d)}\leq C\|f\|_{L^2(\re^d)}+C\|F\|_{L^{2}(-T,T)L^{2_{*},2}(\re^d)}. 
\end{align*}

Replacing $3_*$, $3^*$ in the arguments in Subsection \ref{sttouni} by $2_*$, $2^*$ respectively, we have the following statements:
Let $R_0(z)=(P-z)^{-1}$ for $z\in \mathbb{C}\setminus \s(P)$. Then there exists $C>0$ such that for $z\in \mathbb{C}\setminus \s(P)$ and $f\in L^2(\re^d)\cap L^{2^*,2}(\re^d)$
\begin{align}\label{uniul}
\|R_0(z)f\|_{L^{2^*,2}(\re^d)}\leq C\|f\|_{L^{2_*,2}(\re^d)}.
\end{align}
Moreover, for $w\in L^{d,\infty}(\re^d)$ we have
\begin{align*}
\sup_{z\in \mathbb{C}\setminus \re}\|wR_0(z)w\|_{B(L^2(\re^d))}\leq C\|w\|_{L^{d,\infty}(\re^d)}^2.
\end{align*}
In particular, $\|\jap{x}^{-1}R_0(z)\jap{x}^{-1}\|_{B(L^2(\re^d))}$ is bounded in $z\in \mathbb{C}\setminus \re$.

\section{Resolvent near regular points}

In this section, we study properties of the cut-off resolvent of $H_0$ near regular points of $h_0$.

\begin{lem}\label{aplem}
Let $d\geq 1$ and $\e>0$. Then,
\begin{align}\label{apb1}
\sup_{z\in \mathbb{C}\setminus \re}\|\jap{\y_1}^{-\frac{1}{2}-\e}(D_{\y_1}-z)^{-1}\jap{\y_1}^{-\frac{1}{2}-\e}\|_{B(L^2(\re^d))}<\infty.
\end{align}
If $0<\e\leq 1$, there exists $0<\a_{\e}\leq 1$ such that $\jap{\y_1}^{-\frac{1}{2}-\e}(D_{\y_1}-z)^{-1}\jap{\y_1}^{-\frac{1}{2}-\e}$ is $\a_{\e}$-H\"older continuous in the operator norm topology of $B(L^2(\re^d))$.
\end{lem}

\begin{proof}
Suppose $\Im z> 0$. We denote $\y=(\y_1, \y')$ for $\y\in \re^d$, $\y_1\in\re$ and $\y'\in \re^{d-1}$. Using the Cauchy-Schwarz inequality, we have
\begin{align*}
|(D_{\y_1}-z)^{-1}(\jap{\y}^{-\frac{1}{2}-\e}u(\y))|=&\left|2\pi i\int_{-\infty}^{\y_1}e^{2\pi iz(\y_1-s)}\jap{s}^{-\frac{1}{2}-\e}u(s,\y')ds\right|\\
\leq&C\|u(\cdot, \y')\|_{L^2(\re)}
\end{align*}
for $u\in L^2(\re^d)$, where $C>0$ is independent of $z$, $\y$ and $u$. Thus, we have
\begin{align*}
\int_{\re}|\jap{\y_1}^{-\frac{1}{2}-\e}(D_{\y_1}-z)^{-1}(\jap{\y_1}^{-\frac{1}{2}-\e}u(\y))|^2d\y_1\leq C\|u(\cdot, \y')\|_{L^2(\re)}^2
\end{align*}
with some $C>0$ which is independent of $z$ and $u$.
Integrating the above inequality with respect to $\y'\in\re^{d-1}$, we obtain $($\ref{apb1}$)$. 

Suppose $\Im z, \Im z'>0$. Set $w(\y)=\jap{\y_1}^{-\frac{3}{2}-\e}u(\y)$ for $\e>0$. Using the Taylor theorem and the Cauchy-Schwarz inequality, we have
\begin{align*}
&|\jap{\y_1}^{-\frac{3}{2}-\e}((D_{\y_1}-z)^{-1}-(D_{{\y_1}}-z')^{-1})w(\y)|\\
&=\left|2\pi i\jap{\y_1}^{-\frac{3}{2}-\e}\int_{-\infty}^{\y_1}(e^{2\pi iz(\y_1-s)}-e^{2\pi iz'(\y_1-s)})w(s,\y')ds\right|\\
&\leq2\pi  |z-z'| \jap{\y_1}^{-\frac{3}{2}-\e}\int_{-\infty}^{\y_1} |(\y_1-s)w(s,\y')|ds\\
&\leq C|z-z'|\jap{\y_1}^{-\frac{1}{2}-\e}\|u(\cdot,\y')\|_{L^2(\re)}.
\end{align*}
Integrating the square of the above inequality, we obtain
\begin{align}\label{apb2}
\|\jap{\y_1}^{-\frac{3}{2}-\e}((D_{\y_1}-z)^{-1}-(D_{\y_1}-z')^{-1})\jap{\y_1}^{-\frac{3}{2}-\e}\|_{B(L^2(\re^d))}\leq C|z-z'|.
\end{align}
Moreover, by $($\ref{apb1}$)$, we have
\begin{align}\label{apb3}
\|\jap{\y_1}^{-\frac{1}{2}-\e}((D_{\y_1}-z)^{-1}-(D_{\y_1}-z')^{-1})\jap{\y_1}^{-\frac{1}{2}-\e}\|_{B(L^2(\re^n))}\leq C.
\end{align}

By using a complex interpolation between $($\ref{apb2}$)$ and $($\ref{apb3}$)$, we obtain the H\"older continuity of $\jap{\y_1}^{-\frac{1}{2}-\e}(D_{\y_1}-z)^{-1}\jap{\y_1}^{-\frac{1}{2}-\e}$ in $B(L^2(\re^d))$ for $\e>0$. The case $\Im z<0$ is similarly proved.

\end{proof}

For a proof of our main result in this section, we need the following two lemmas.
\begin{lem}\label{lemb2}
Let $\chi, \g\in C_c^{\infty}(\re^d)$ satisfy $\supp\chi\subset \{\g=1\}$. Then, for $\a\in \re$ there exits $C>0$ such that
\begin{align*}
\|(1-\g)\jap{D}^{\a}\chi u\|_{L^2(\re^d)}\leq C\|u\|_{L^2(\re^d)},\quad u\in L^2(\re^d).
\end{align*}

\end{lem}
\begin{proof}
This lemma follows from the disjoint support property of pseudodifferential operators. For the sake of the completeness of this paper, we give a self-contained proof. Considering the support property of $\chi$ and $\g$, we observe that $c|x|\leq |x-y|\leq C|x|$ on $\supp (1-\g(x))\chi(y)$. Set $L=(1+|x-y|^2)^{-1}(1-(x-y)\cdot D_{\x})$, then note that $Le^{2\pi i(x-y)\cdot \x}=e^{2\pi i(x-y)\cdot \x}$. Integrating by parts, we have
\begin{align*}
&|(1-\g(x))\jap{D}^{\a}\chi u(x)|\\
=&\left|(1-\g(x))\int_{\re^{2d}}(L^*)^N(\jap{\x}^{\a})e^{2\pi i(x-y)\cdot \x}(\chi u)(y)dyd\x\right|\\
\leq&|1-\g(x)|\int_{\re^d}\frac{1}{|x-y|^{d+1}}|\chi(y)u(y)|dy\\
\leq&\jap{x}^{-2-2d}\|u\|_{L^2(\re^d)}
\end{align*}
for any integer $N>\a+d+1$. Integrating the square of the above inequality with respect to $x\in \re^d$, we obtain the desired result.
\end{proof}

\begin{lem}\label{lemb3}
Let $U\subset \T^d$ be an open set and $\k$ be a diffeomorphism from $U$ onto an open set in $\re^d$. Set $u_{\k}(\y)=u(\k^{-1}(\y))$. Then, for $\chi\in C^{\infty}_c(U)$ and $\a\geq 0$, we have
\begin{align*}
\|\chi_{\k} u_{\k}\|_{H^{\a}(\re^d)}\leq C\|u\|_{H^{\a}(\T^d)},\quad u\in H^{\a}(\T^d)
\end{align*}
for some $C>0$.
\end{lem}
\begin{proof}
We take $\f,\g\in C_c^{\infty}(U)$ satisfying $\supp \chi\subset \{\f=1\}$ and $\supp \f\subset \{\g=1\}$. Then we have
\begin{align*}
\|\chi_{\k}u_{\k}\|_{H^{\a}(\re^d)}\leq \|\g_{\k}\jap{D}^{\a}\chi_{\k}u_{\k}\|_{L^2(\re^d)}+\|(1-\g_{\k})\jap{D}^{\a}\f_{\k}\chi_{\k}u_{\k}\|_{L^2(\re^d)}. 
\end{align*}
Using Lemma \ref{lemb2}, we learn
\begin{align*}
\|(1-\g_{\k})\jap{D}^{\a}\f_{\k}\chi_{\k}u_{\k}\|_{L^2(\re^d)}\leq C\|\chi_{\k}u_{\k}\|_{L^2(\re^d)}\leq C\|u\|_{L^2(\T^d)}.
\end{align*}
Due to the coordinate invariance of the Sobolev spaces and the support property of $\g_{\k}$, we obtain
\begin{align*}
 \|\g_{\k}\jap{D}^{\a}\chi_{\k}u_{\k}\|_{L^2(\re^d)}\leq C\|u\|_{H^{\a}(\T^d)}.
\end{align*}
This completes the proof.
\end{proof}

\begin{rem}
The above lemma is trivial if $2\a$ is an integer. The difficulty is due to the lack of the local property of the pseudodifferental operator $\jap{D}^{2\a}$ if $2\a$ is not an integer. 
\end{rem}

We now state the main result of this section.

\begin{prop}\label{cutMo}
Suppose $d\geq 1.$ Let $\chi\in C^{\infty}(\T^d)$ be a real-valued function satisfying $\supp\chi \subset \{\nabla h_0\neq 0\}$. Then,
\begin{align*}
\sup_{z\in \mathbb{C}\setminus \re}\|\jap{x}^{-\frac{1}{2}-\e}\chi(D)(H_0-z)^{-1}\chi(D)\jap{x}^{-\frac{1}{2}-\e}\|_{B(\H)}<\infty.
\end{align*}
Moreover, $\jap{x}^{-\frac{1}{2}-\e}\chi(D)(H_0-z)^{-1}\chi(D)\jap{x}^{-\frac{1}{2}-\e}$ is $\a_{\e}$-H\"older continuous in the operator norm topology of $B(\H)$, where $\a_{\e}$ is the constant in Lemma \ref{aplem}.
\end{prop}

\begin{proof}
By using a partition of unity, we may suppose that $\supp\chi$ is small enough. Thus, we may suppose   $\pa_{\x_1}h_0(\x)\neq 0$ on $\supp \chi$ without loss of generality. Set $\y=\k(\x)=(h_0(\x),\x')$. Then the inverse function theorem implies that $\k$ is a diffeomorphism from a neighborhood of $\supp \chi$ onto its image. We denote $\k^{-1}(\y)=(\x_1(\y),\y')$ for $\y\in \k(\supp \chi)$. We also denote $f_{\k}(\y)=f(\k^{-1}(\y))$. Using Lemma \ref{aplem} and Lemma \ref{lemb3}, we have
\begin{align*}
&\left|\int_{\T^d}\bar{f}(\x)\chi(\x)^2g(\x)(h_0(\x)-z)^{-1}d\x\right|\\
&=\left|\int_{\re^d}\bar{f}_{\k}(\y)\chi_{\k}(\y)^2g_{\k}(\y)(\y_1-z)^{-1}\frac{d\y}{|(\pa_{\x_1}h_0)(\x_1(\y),\y')|}\right|\\
&\leq C\|\chi_{\k}f_{\k}\|_{H^{\frac{1}{2}+\e}(\re^d)}\|\chi_{\k}g_{\k}\|_{H^{\frac{1}{2}+\e}(\re^d)}\\
&\leq C\|f\|_{H^{\frac{1}{2}+\e}(\T^d)}\|g\|_{H^{\frac{1}{2}+\e}(\T^d)}.
\end{align*}
Similarly, we have
\begin{align*}
&\left|\int_{\T^d}\bar{f}(\x)\chi(\x)^2g(\x)((h_0(\x)-z)^{-1}-(h_0(\x)-z')^{-1})d\x\right|\\
&\leq C|z-z'|^{\a_{\e}}\|f\|_{H^{\frac{1}{2}+\e}(\T^d)}\|g\|_{H^{\frac{1}{2}+\e}(\T^d)}.
\end{align*}
By using the Fourier transform, these imply the desired results.
\end{proof}

\end{document}